\documentclass[10pt,conference]{IEEEtran}
\IEEEoverridecommandlockouts
\usepackage[dvipdf]{graphicx,color}

\usepackage{amssymb}
\usepackage{amsmath}
\usepackage{stfloats}
\usepackage{amsfonts}
\usepackage{balance}
\usepackage{color}
\usepackage{algorithm}
\usepackage{algpseudocode}
\usepackage{color}
\usepackage{varwidth}
\usepackage{multicol}
\usepackage{subfigure}
\usepackage{xspace}
\usepackage{enumerate}
\usepackage{xcolor,cite,etoolbox}
\usepackage{bm}
\usepackage{amsthm}
\usepackage{hyperref}
 \usepackage{caption}
\usepackage{subcaption}
\newtheorem{proposition}{Proposition}
\newtheorem{corollary}{Corollary}
\newtheorem{remark}{Remark}
\usepackage[font=footnotesize,labelfont=footnotesize]{caption}
\usepackage{ragged2e}

\def\BibTeX{{\rm B\kern-.05em{\sc i\kern-.025em b}\kern-.08em
    T\kern-.1667em\lower.7ex\hbox{E}\kern-.125emX}}

\begin{document}

\title{Outage Analysis of RSMA-Enabled Terrestrial Users Under
Co-Channel Interference from HAPS}

\author{\IEEEauthorblockN{Farjam Karim\IEEEauthorrefmark{1},   Nurul Huda Mahmood\IEEEauthorrefmark{1}, Prathapasinghe Dharmawansa\IEEEauthorrefmark{1},  Deepak Kumar\IEEEauthorrefmark{2},  and Matti Latva-aho\IEEEauthorrefmark{1}%
	\thanks{This research was supported by Interreg Aurora Project. }}\\
\IEEEauthorblockA{\IEEEauthorrefmark{1}Centre for Wireless Communications, University of Oulu, Finland. \\ 
\IEEEauthorblockA{\IEEEauthorrefmark{2}Department of Electronics and Communication Engineering, MNNIT, Allahabad, Prayagraj 211004, India. \\
  }
		Email: \{farjam.karim,\;Prathapasinghe.KaluwaDevage,\;nurulhuda.mahmood,\;matti.latva-aho\}@oulu.fi,\\
        dkumar@mnnit.ac.in 
		}} 

\maketitle

\begin{abstract}
The coexistence of terrestrial networks (TNs) with high-altitude platform station (HAPS)-based non-terrestrial networks (NTNs) is a promising approach for extending 6G connectivity, but the resulting cross-network interference can significantly affect TN reliability. This paper investigates an rate-splitting multiple access (RSMA)-based TN operating under interference from a multibeam HAPS-NTN system. The desired TN links are modeled using Nakagami-$m$ fading, while the aggregate HAPS interference is characterized by shadowed-Rician fading. Closed-form outage probability expressions are derived by characterizing the aggregate interference and accounting for both perfect and imperfect successive interference cancellation (SIC). Monte Carlo simulations are used to verify the analytical results and to examine the effects of HAPS transmit power, the number of interfering beams, and SIC imperfections. The results demonstrate that the interference-to-signal power scaling is a key factor determining TN outage behavior, with jointly scaled HAPS interference leading to interference-limited performance. Moreover, RSMA provides improved outage performance compared with non-orthogonal multiple access (NOMA), particularly in the presence of HAPS interference. 
\end{abstract}


\section{Introduction}
The sixth generation (6G) of wireless networks is expected to evolve into a unified space-air-ground integrated network, in which satellites, high-altitude platform stations (HAPS), unmanned aerial vehicles, and terrestrial base stations (BSs) jointly deliver seamless connectivity~\cite{yue_xiao_jsac_2024}. This direction is reflected in the ITU-R IMT-2030 framework and ongoing 3GPP non-terrestrial network (NTN) standardization efforts, both of which treat non-terrestrial platforms as native components of future radio access networks~\cite{itu_m2160, 3gpp_38811}. HAPS, operating at an altitude of roughly 20 kilometers (km), occupies a particularly attractive position within this architecture: it offers markedly lower propagation delay and path loss than low-Earth-orbit satellites. Its multi-beam antenna also enables a single platform to cover an area far larger than any individual terrestrial BS, making it an attractive means of extending broadband access to rural, maritime, and disaster-affected regions~\cite{abbasi_wire_comm_2024, khan_std_mag_2026, karaman_iot_mag_2026}.

To meet growing spectral-efficiency demands, HAPS is expected to reuse licensed terrestrial spectrum~\cite{haps_alliance_ntn_2024}. Such terrestrial network (TN)-NTN spectrum sharing improves resource utilization but exposes terrestrial users to co-channel interference from the HAPS. This effect is compounded by the fact that a HAPS payload typically activates multiple beams simultaneously to serve its own user population, so a terrestrial user within overlapping footprints experiences interference from several HAPS beams rather than one. 

Rate-splitting multiple access (RSMA) has recently emerged as a flexible multiple access strategy capable of efficiently managing interference. In RSMA, each user's message is divided into a common part and a private part. The common parts of all users are jointly encoded into a single common stream decoded by every user, whereas the private parts are independently encoded into private streams intended for their respective users~\cite{Bruno_proc, Farjam_twc_2025}. After decoding and removing the common stream through successive interference cancellation (SIC), each user decodes its own private stream while treating the remaining interference as noise. By enabling interference to be partially decoded and partially treated as noise, RSMA is well suited to terrestrial users sharing spectrum with a HAPS, where interference originates from independently scheduled multi-beam transmissions and cannot be coordinated. Although downlink RSMA (see, e.g., Section II-E of~\cite{Bruno_proc} and the references therein) and HAPS-enabled communications~\cite{abbasi_wire_comm_2024, khan_std_mag_2026, karaman_iot_mag_2026} have each been investigated upto certain extent, their integration under realistic multi-beam interference remains largely unexplored. As a result, the outage performance of terrestrial RSMA users under realistic multibeam HAPS interference has not, to the best of our knowledge, been analytically characterized.

This gap motivates our work. As HAPS systems scale up their beam count to
serve larger populations, the aggregate interference experienced by an
affected terrestrial user grows in ways that single-beam analyses cannot
capture, and RSMA's ability to absorb that growth remains unquantified. We
address this gap by developing an analytical framework for an RSMA-enabled
TN-NTN network in which a terrestrial BS and a multi-beam HAPS operate
concurrently over shared spectrum. The main contributions of this work are
as follows:
\begin{itemize}
    \item We formulate an RSMA-enabled TN-NTN coexistence model in which
    terrestrial users experience aggregate co-channel interference from
    multiple HAPS beams, jointly capturing Nakagami-$m$ fading on the
    desired TN link and shadowed-Rician fading on the interfering NTN link.
    \item We approximate the aggregate multi-beam interference via a
    moment-matched Gamma distribution and derive a closed-form outage
    probability expression (TN link) for the  RSMA common and private streams,
    explicitly parameterized by the number of interfering HAPS beams.
    \item We validate the analysis against Monte Carlo simulation and
    provide design insights into how RSMA power allocation and HAPS beam
    count jointly affect terrestrial outage performance.
\end{itemize}

\noindent \textbf{Notations:}
$\mathcal{CN}(0,\sigma^2_{(\cdot)})$ represents a zero-mean
circularly symmetric complex Gaussian distribution with variance $\sigma^2_{(\cdot)}$, while $\mathbb{E}[\cdot]$ and $|\cdot|$ stand for the expectation and modulus operators, respectively. The probability density function (PDF) and cumulative distribution function (CDF) of a random variable $X$ are written as $f_X(\cdot)$ and $F_X(\cdot)$, respectively, and its complementary CDF is defined as $\bar F_X(\cdot)\triangleq 1-F_X(\cdot)$. The symbols $\Gamma(\cdot)$,
$\exp(\cdot)$, and $\left(!\right)$ denote the Gamma function, the exponential function, and the factorial operator, respectively, and $\binom{\cdot}{\cdot}$
represents the binomial coefficient. 

\section{System Model}

We consider an integrated terrestrial-non-terrestrial (TN-NTN) spectrum-sharing network, where a terrestrial cellular system and a high-altitude platform station (HAPS) simultaneously operate over the same frequency band as illustrated in Fig.~\ref{system_mod_fig}. The terrestrial network consists of a single-antenna base station (BS) serving a set of  $\mathcal{K}=\{1,2, k,\cdots,K\}$ terrestrial users,
using downlink RSMA. Concurrently, the HAPS employs $\mathcal{N}=\{1,2, n,\cdots,N\}$ multiple-beam and also utilizes RSMA to serve its own $\mathcal{V}=\{1,2, v,\cdots,V\}$  users in a single beam over the same radio resources.

Since the terrestrial and HAPS systems reuse identical time-frequency resources, co-channel interference may arise whenever a terrestrial user lies within the coverage of an active HAPS beam. Consequently, the terrestrial users are divided into two categories. The first category consists of users that are outside the footprint of all HAPS beams and therefore receive only the desired signal transmitted by the terrestrial BS. The second category consists of users that fall within the coverage of one or more HAPS beams and consequently receive both the desired terrestrial transmission and co-channel interference from the HAPS. The primary objective of this work is to analyze the outage performance of the latter class of users.

Owing to the distinct propagation environments, the desired terrestrial and interfering HAPS links are modeled differently. The terrestrial BS is equipped with a single transmit antenna and communicates with $\mathcal{K}$ single antenna users over  terrestrial propagation channels. Accordingly, the desired BS-user links experience large-scale terrestrial path loss together with small-scale fading which assumed to be Nakagami-$m$ fading for this work. In contrast, the interfering HAPS links originate from multi-beam transmissions and therefore incorporate free-space propagation loss, atmospheric gaseous absorption, and other high-altitude propagation impairments in addition to shadowed-Rician fading.

The channel models for the desired terrestrial transmission and the interfering HAPS transmission are discussed in the following subsection.
\begin{figure}[t!]
    \centering
\includegraphics[width=0.95\linewidth]{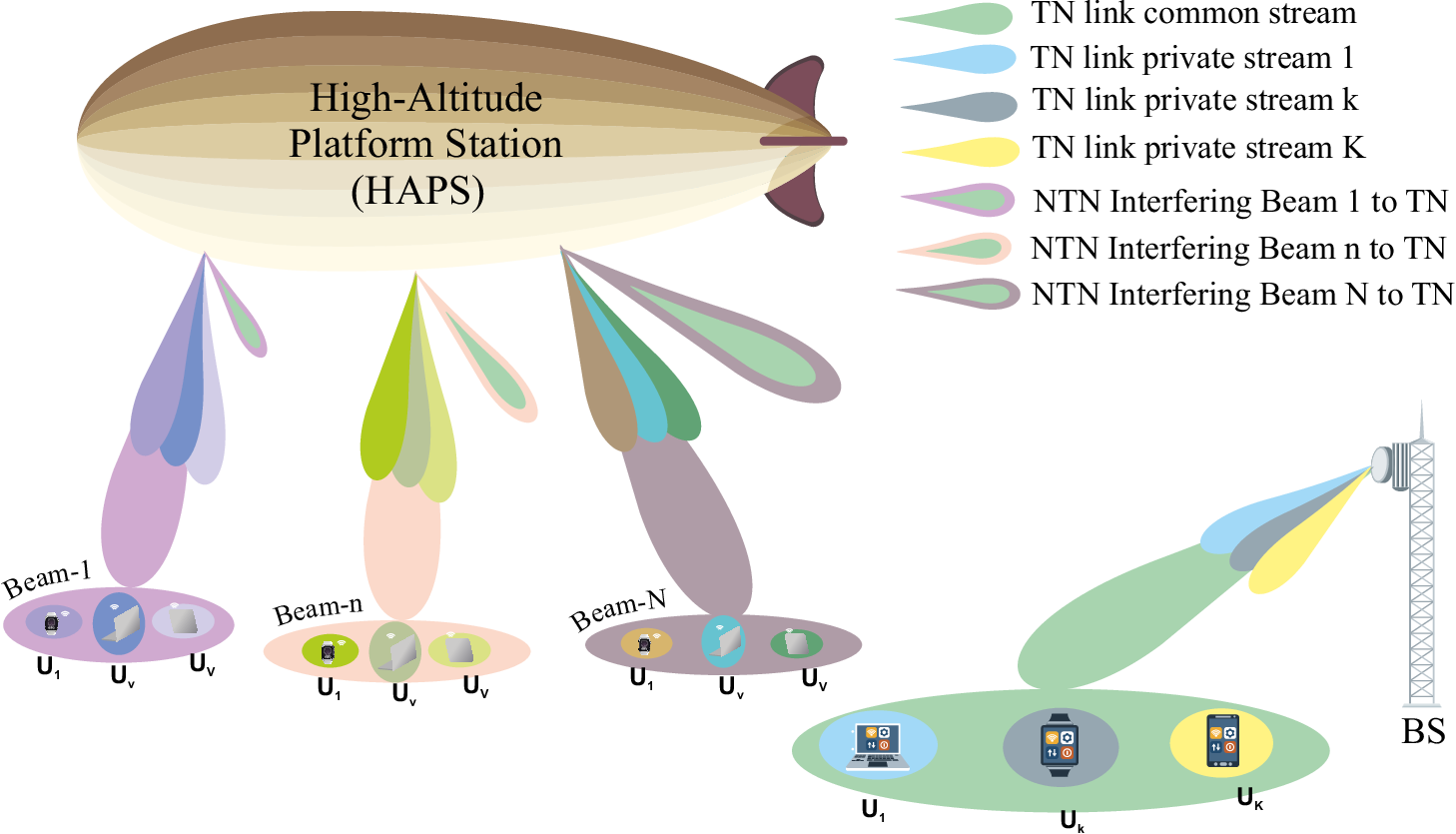}
    \caption{A schematic of the system model.}
    \label{system_mod_fig}
    \vspace{-1em}
\end{figure}
\subsection{Desired terrestrial link: BS $\rightarrow U^{(T)}_k$}

The desired terrestrial channel is modeled as a product of a deterministic large-scale channel gain and a small-scale fading coefficient. Accordingly, the complex channel coefficient between the terrestrial BS and the $k$-th user $U^{(T)}_k$ is given as
\begin{equation}\label{first_eq}
h_k=\sqrt{L_k}g_k, \;X_k\triangleq |h_k|^2,
\end{equation}
where $g_k$ denotes the small-scale Nakagami-$m$ fading coefficient with integer fading parameter $m_T$. The fading coefficient is normalized as $\mathbb{E}\left[|g_k|^2\right]=1$,
such that the complete average channel power is captured by the deterministic coefficient $L_k$, whereas $g_k$ represents only the random small-scale fluctuations around this average.
The large-scale gain of the terrestrial BS-$U^{(T)}_k$ link is expressed as
$L_k = G_{\mathrm{BS}}\, G_{U_k}^{(t)} 10^{-\mathrm{PL}_k/10}$,
where $G_{\mathrm{BS}}$ and $G_{U_k}^{(t)}$ denote the BS transmit antenna
gain and terrestrial-user receive antenna gain, respectively, and
$\mathrm{PL}_k$ [dB] is the free-space path loss over the BS-$U^{(T)}_k$
distance $d_k$, given by
    $\mathrm{PL}_k = 92.45 + 20\log_{10}\!\left(f_c[\mathrm{GHz}]\right)
    + 20\log_{10}\!\left(d_k[\mathrm{km}]\right)$. Since the terrestrial communication occurs over a relatively short propagation distance (compared to HAPS), atmospheric effects such as gaseous absorption, rain attenuation, and cloud attenuation are neglected.
The average channel power gain for the TN is thus, obtained as
$\Omega_k\triangleq\mathbb{E}[X_k]=\mathbb{E}[L_k|g_k|^2]=L_k\mathbb{E}[|g_k|^2]=L_k$.
 Consequently, the channel power gain follows a Gamma distribution
with its CDF given as~\cite{Farjam_twc_2025}:
\begin{align}  \label{CDF}
F_{X_k}(x)= \frac{1}{\Gamma\left(m_k\right)} \gamma\left(m_k, \frac{ m_k }{L_k}x\right), \qquad x\geq 0.
\end{align}
For $m_k=1$, the Nakagami-$m$ reduces to Rayleigh fading, and $X_k$ becomes exponentially distributed.

\subsection{Interference channel: HAPS beam $n\rightarrow U^{(T)}_k$}
\label{sec:SRfadng}
Unlike the desired terrestrial BS-$U^{(T)}_k$ link, the interference link from the HAPS $n$-{th} beam  to the terrestrial user $U^{(T)}_k$ experiences a long-distance air-to-ground propagation environment. In particular, the HAPS transmission experiences beam-dependent antenna gain, large propagation distance, and atmospheric attenuation effects, including free-space spreading loss, gaseous absorption, and rain attenuation. Furthermore, owing to the presence of a dominant line-of-sight (LoS) component together with random shadowing and multipath scattering, the small-scale fading is modeled by the shadowed-Rician distribution~\cite{yun_ai_photon_2019,Yahia_aerospace_2022}:
\begin{align}
g^{(H)}_{k, n}=\sqrt{C_{k,n}}\,\tilde g_{k,n}, \qquad Z_{k,n}\triangleq|\tilde g_{k,n}|^2,
\end{align}
The large-scale gain of this link can be expressed as
\begin{align}
C_{k,n} = \frac{G_H(\varphi_{k,n}) G_{U_k}^{(t)}}{L_F(d_{k,n})\,L_a\,L_r(d_{k,n})},\label{eq}
\end{align}
where $G_H(\varphi_{k,n})$ is the HAPS transmit antenna gain of the $n$-th beam at off-boresight angle $\varphi_{k,n}$ toward $U_k^{(T)}$; and $L_F(d_{k,n})$, $L_a$, and $L_r(d_{k,n})$ are the free-space, atmospheric, and rain attenuation losses, respectively (all in dB).
The free-space propagation loss in dB is evaluated as $L_F^{\mathrm{dB}}=92.45+20\log_{10}\!\big(f_c[\mathrm{GHz}]\big) + 20\log_{10}\!\big(d_{k,n}[\mathrm{km}]\big)$, where $f_c$ is the carrier frequency and $d_{k,n}=\mathcal{H}/\sin(\psi_{k,n})$, with $\mathcal{H}$ the HAPS altitude and $\psi_{k,n}$ denotes the elevation angle from $U_k^{(T)}$ to the $n_{\text{th}}$ HAPS beam. The corresponding linear-scale loss used in \eqref{eq} is $L_F=10^{L_F^{\mathrm{dB}}/10}$. $G_H(\varphi_{k,n})$ follows the reference pattern
$G_H^{\mathrm{dB}}(\varphi_{k,n}) = G_H^{\max} - \min\!\left(12\left(\frac{\varphi_{k,n}}{\varphi_{3\mathrm{dB}}}\right)^2,\, A_m\right)$,
where $G_H^{\max}$, $\varphi_{3\mathrm{dB}}$, and $A_m$ are the maximum boresight gain, 3-dB beamwidth, and maximum sidelobe attenuation, respectively. Assuming a nadir-pointing boresight, $\varphi_{k,n}=90^\circ-\psi_{k,n}$.
Unlike the terrestrial link, these attenuation factors cannot be neglected because of the long slant range between the HAPS and the terrestrial user. 
As $Z_{k,n}$ is assumed to follow the
shadowed-Rician fading; its PDF for integer shadowing severity $m_{k,n}$ is thus expressed as:
\begin{align}
f_{Z_{k,n}}(z)=\alpha_{k,n}\sum_{r=0}^{m_{k,n}-1}\binom{m_{k,n}-1}{r}\frac{\delta_{k,n}^{r}z^r}{r!}\nonumber\\
\times \exp\left({-(\beta_{k,n}-\delta_{k,n})z}\right),
\end{align}
Smaller values of $m_{k,n}$ correspond to severe shadowing, whereas larger
values indicate milder shadowing, with the model approaching the conventional
(unshadowed) Rician fading distribution as $m_{k,n}\to\infty$. The parameters $\alpha_{k,n}$, $\beta_{k,n}$, and $\delta_{k,n}$ depend on the average LoS power $\Omega_{k,n}$ and the average scattered-component power $2b_{k,n}$, and are expressed as
$\alpha_{k,n}=\frac{1}{2b_{k,n}}\left(
\frac{2b_{k,n} m_{k,n}}{2b_{k,n} m_{k,n}+\Omega_{k,n}}
\right)^{m_{k,n}},
\beta_{k,n}=\frac{1}{2b_{k,n}},
\delta_{k,n}=\frac{\Omega_{k,n}}
{2b_{k,n}(2b_{k,n} m_{k,n}+\Omega_{k,n})}$.
 
 The mean power of the shadowed-Rician fading component itself is the sum of the average LoS power and the average multipath power, which can be written as
$\mathbb{E}[Z_{k,n}] = \Omega_{k,n}+2b_{k,n}$.
Therefore, $C_{k,n}$ represents only the deterministic propagation and antenna-gain contribution, whereas $\Omega_{k,n}$, $2b_{k,n}$, and $m_{k,n}$ characterize the physical shadowing and multipath statistics of the HAPS channel to $U^{(T)}_k$.
Taking expectation, the average
received interference power from HAPS $n$-th beam can be given as
\begin{align}
\mathbb{E}[Y^{(H)}_{k,n}] \!=\! C_{k,n}P_{H,n}\mathbb{E}[Z_{k,n}]\! = \!C_{k,n}P_{H,n}\!\left(\Omega_{k,n}\!+\!2b_{k,n}\right),
\label{eq:Ym_mean}
\end{align}
where $P_{H,n}$ is transmit power budget avaliable for the $n$-{th} beam.
Having established the propagation models for the desired terrestrial and interfering HAPS links, we next present the transmitted signal models of the TN and HAPS systems, followed by the received signal at the terrestrial user and the corresponding SINR expressions.
\subsection{Received signal and SINR}
At the TN transmitter side each user's message is split by the rate-splitting encoder into two parts: a common part and a private part~\cite{Bruno_proc, Farjam_twc_2025}. The common parts of all users are combined and encoded into a single common stream $s_{T,c}$, drawn from a shared codebook that all the users are capable of decoding, while each user's private part is encoded into an individual private stream, denoted $s_{T,k}$ for $U^{(T)}_k$. These streams are then linearly superposed in the power domain to form the transmitted signal:
$x_b = \sqrt{P_T \zeta_{T,c}}s_{T,c}+ \sum\limits^{K}_{k=1}\sqrt{P_T \zeta_{T,k}} s_{T,k}$,
where $\mathbb{E}[|s_{T,c}|^2]=\mathbb{E}[|s_{T,k}|^2]=1$, and $P_T$ denotes the power budget available for transmitting the signal $x_b$, whereas $\zeta_{T,c}$ and $\zeta_{T,k}$ denotes the power allocation for common stream and the $U^{(T)}_k$ private stream, respectively. Note that $\zeta_{T,c} + \sum\limits^{K}_{k=1}\zeta_{T,k}=1$.
Moreover, the HAPS $n$-{th} beam also uses the same RSMA principles to serve its own set of users and its transmitted RSMA signal can be expressed as $x_{H,n} = \sqrt{P_{H,n} \zeta^{(H)}_{n,c}}s^{(H)}_{n,c} + \sum\limits^{V}_{v=1}\sqrt{P_{H,n} \zeta^{(H)}_{v,n}} s^{(H)}_{v,n}$. The notation for the HAPS transmitted signal follows the same RSMA
transmission principle as that of the TN. Although the HAPS also employs
RSMA, $U^{(T)}_k$ does not possess the HAPS codebooks and therefore cannot
decode or perform SIC on the HAPS streams. Consequently, the entire received
HAPS signal acts as co-channel interference.

Let $\mathcal{N}_k\subseteq\mathcal{N}$ denote the set of HAPS beams covering
the terrestrial user $U^{(T)}_k$. The instantaneous interference power
received from the $n$-th HAPS beam is
\begin{equation}
Y^{(H)}_{k,n} \triangleq C_{k,n}P_{H,n}Z_{k,n}.
\end{equation}
Hence, the aggregate HAPS interference at $U^{(T)}_k$ is
$I^{(H)}_{k}=\sum_{n\in\mathcal{N}_k}Y^{(H)}_{k,n}$.
The received signal at the terrestrial user $U^{(T)}_k$ is given by
\begin{align}
y^{(T)}_k=h_k\left(\sqrt{P_T \zeta_{T,c}}s_{T,c}
+ \sum\limits^{K}_{k=1}\sqrt{P_T \zeta_{T,k}} s_{T,k}\right)\nonumber\\+\sum_{n\in\mathcal N_k}g^{(H)}_{k,n}x_{H,n}+w_k,
\label{eq:received_signal}
\end{align}
where $w_k\sim\mathcal{CN}(0,\sigma_{T,k}^2)$ is the additive white Gaussian noise.
According to the RSMA decoding procedure, $U^{(T)}_k$ first decodes the common
stream while treating all private streams and the aggregate HAPS interference
as noise. Therefore, the corresponding  SINR for the common stream can be given as
\begin{align}\label{comm_sinr}
\Upsilon^{(T)}_{c,k}=\frac{P_T\zeta_{T,c}X_k}{P_T(1-\zeta_{T,c})X_k+I^{(H)}_k+\sigma_{T,k}^2},
\end{align}
where $X_k$ is given in \eqref{first_eq}. After decoding the common stream, $U^{(T)}_k$ performs SIC before decoding its own private stream. To account for imperfect SIC, let $0\leq\varepsilon\leq1$ denote the residual interference factor, where $\varepsilon=0$ corresponds to perfect SIC case. Consequently, the SINR for decoding the private stream can be given as
\begin{align}
\Upsilon^{(T)}_{p,k}=\frac{P_T\zeta_{T,k}X_k}{P_T\left(\sum\limits_{i=1, i\neq k}^{K}\zeta_{T,i}+\varepsilon\zeta_{T,c}\right)X_k+I_k^{(H)}+\sigma_{T,k}^2 }.
\label{SINR_private}
\end{align}These SINR expressions serve as the basis for the performance analysis 
presented in the following section.
\section{Performance Analysis}\label{sec:analysis}
In this section, we derive the outage probability and throughput expressions of the considered system for two interference scenarios: (i) a single interfering HAPS beam affecting $U^{(T)}_k$, and (ii) multiple interfering HAPS beams simultaneously affecting $U^{(T)}_k$.  Let  $\Upsilon^{th}_{c,k}=2^{R^{(T)}_{c}}-1$ and $\Upsilon^{th}_{p,k}=2^{R^{(T)}_{p,k}}-1$ denote the target SINR thresholds for $U^{(T)}_k$ common and private stream, respectively with $R^{(T)}_{c}$ and $R^{(T)}_{p,k}$ representing the corresponding target rates. 
 An outage occurs whenever either the common stream or the private
stream cannot be successfully decoded. If \eqref{comm_sinr} and \eqref{SINR_private} fails to cross the thresholds $\Upsilon^{th}_{c,k}$ and $\Upsilon^{th}_{p,k}$, respectively, then $U^{(T)}_k$ will be in an outage.
The respective outage probability considering single interfering HAPS beam is evaluated in the following proposition.
\begin{proposition}
The outage probability of $U^{(T)}_k$ under a single interfering HAPS beam is given by
    \begin{align}\label{single_outage}
        \mathrm{P}_{\text{out},k} = &1-\sum\limits_{r=0}^{m_{k, n}-1}\sum\limits_{\ell=0}^{m_k-1}\sum\limits_{p=0}^{\ell}\binom{m_{k, n}-1}{r}\binom{\ell}{p}\left(\sigma^2_{T,k}\right)^{\ell-p}\nonumber\\
        &\times\frac{\alpha_{k, n}\delta^r_{k, n}\exp\left(-\varphi_{k}\sigma^2_{T,k}\right)\varphi_{k}^\ell\left(r+p\right)!}{r!\; \ell!\;B^{r+1}_{k,n}\left(\lambda_{k,n}+\varphi_{k}\right)^{r+p+1}},
    \end{align}
    where $\varphi_{k}= \frac{m_k\xi_k}{L_{k}}$,  $\lambda_{k,n}=\frac{\beta_{k,n}-\delta_{k,n}}{B_{k,n}}$, $B_{k,n}=C_{k,n}P_{H,n}$, $\xi_k = \max\left(\xi_{c,k}, \xi_{p,k}\right)$, $\xi_{c,k}=\frac{\Upsilon^{th}_{c,k}}{P_T\left(\zeta_{T,c}-\Upsilon^{th}_{c,k}(1-\zeta_{T,c})\right)}$, and $\xi_{p,k}=\frac{\Upsilon^{th}_{p,k}}{P_T\left(\zeta_{T,k}-\Upsilon^{th}_{p,k}\left(\sum\limits_{i=1, i\neq k}^{K}\zeta_{T,i}+\varepsilon\zeta_{T,c}\right)\right)}$, with $\varepsilon$ denoting the residual interference factor due to imperfect SIC. The case of perfect SIC is obtained by setting $\varepsilon=0$. Note that, \eqref{single_outage} is valid provided  $\zeta_{T,c}> \frac{\Upsilon^{th}_{c,k}}{1+\Upsilon^{th}_{c,k}}$ and $\zeta_{T,k}>\Upsilon^{th}_{p,k}\left(\sum\limits_{i=1, i\neq k}^{K}\zeta_{T,i}+\varepsilon\zeta_{T,c}\right)$, otherwise $ \mathrm{P}_{\text{out},k}=1$.
\end{proposition}
\begin{proof}
    By definition, $U^{(T)}_k$ is in outage unless both the common and the private
stream are successfully decoded, i.e.
\begin{equation}
 \mathrm{P}_{\text{out},k} = 1-\Pr\Big(\Upsilon^{(T)}_{c,k}\geq\Upsilon^{th}_{c,k},\;\Upsilon^{(T)}_{p,k}\geq\Upsilon^{th}_{p,k}\Big).
\label{eqA:def}
\end{equation}
Using \eqref{comm_sinr} and \eqref{SINR_private}, the joint decoding event
$\{\Upsilon^{(T)}_{c,k}\ge\Upsilon^{th}_{c,k},\Upsilon^{(T)}_{p,k}\ge\Upsilon^{th}_{p,k}\}$
reduces, after collecting the $X_k$ terms on each SINR constraint, to the
single condition
\begin{equation}
X_k\ge\varphi_k\big(I_k^{(H)}+\sigma_{T,k}^2\big),
\end{equation}
with $\varphi_k=\max(\xi_{c,k},\xi_{p,k})$ as given in the proposition,
provided the feasibility conditions $\zeta_{T,c}>\Upsilon^{th}_{c,k}/(1+\Upsilon^{th}_{c,k})$
and $\zeta_{T,k}>\Upsilon^{th}_{p,k}\big(\sum_{i\ne k}\zeta_{T,i}+\varepsilon\zeta_{T,c}\big)$
hold. Since $X_k$ and $I_k^{(H)}=B_{k,n}Z_{k,n}$ are independent,
\begin{align}\label{main_int}
 \mathrm{P}_{\text{out},k} = 1-\int_0^\infty f_{Y^{(H)}_{k,n}}(y)\,\bar F_{X_k}\big(\varphi_k(y+\sigma_{T,k}^2)\big)\,dy.
\end{align}
For integer $m_k$,
$\bar F_{X_k}(x)=\exp\left({-\frac{m_kx}{L_k}}\right)\sum_{\ell=0}^{m_k-1}\frac{1}{\ell!}\left(\frac{m_kx}{L_k}\right)^\ell$.
(Erlang form of \eqref{CDF}), and a change of variables $z=y/B_{k,n}$ in the
shadowed-Rician PDF of $Z_{k,n}$ gives $I_k^{(H)}$'s PDF directly as the finite Gamma mixture
\begin{align}
    f_{Y^{(H)}_{k,n}}(y)\!=\!\!\!\!\!\!\sum_{r=0}^{m_{k,n}-1}\!\!\!\alpha_{k,n}\binom{m_{k,n}\!-\!1}{r}\frac{\delta^r_{k,n}y^r}{r!B_{k,n}^{r+1}}\exp\left({-\lambda_{k,n}y}\right)
\end{align}Substituting both, expanding $(y+\sigma_{T,k}^2)^\ell$ via the binomial theorem, and evaluating the resulting integral with
$\int_0^\infty y^{r+p}\exp({-\eta y})\; dy=(r+p)!/\eta^{r+p+1}$, where
$\eta=(\lambda_{k,n}+\varphi_k)$, yields \eqref{single_outage}. 
\end{proof}
\begin{corollary}
When no HAPS beam interferes with $U_k^{(T)}$, $I_k^{(H)}\equiv 0$, so the decoding condition $X_k\ge\varphi_k\big(I_k^{(H)}+\sigma^2_{T,k}\big)$ reduces directly to $X_k\ge\varphi_k\sigma^2_{T,k}$. Evaluating \eqref{eqA:def} with the Erlang complementary CDF of $X_k$ then yields the interference-free outage probability
$\mathrm{P}_{\text{out},k}^{(0)} = 1-\exp\!\left(-\frac{m_k\varphi_k\sigma^2_{T,k}}{L_k}\right)\sum_{\ell=0}^{m_k-1}\frac{1}{\ell!}\left(\frac{m_k\varphi_k\sigma^2_{T,k}}{L_k}\right)^\ell$.
\end{corollary}
\begin{remark}
  The throughput for $U^{(T)}_k$ under a single interfering HAPS beam can be expressed as
    \begin{align}\label{throu_1}
        \mathcal{T}_{k}
= \left(1-\mathrm{P}_{\text{out},k} \right)\left({R^{(T)}_{c}}+{R^{(T)}_{p,k}}\right).   \end{align}
By substituting the expression obtained using Corollary~1, the throughput for $U^{(T)}_k$ under a single interfering HAPS beam can be obtained directly.
\end{remark}
Next, we consider the case where multiple HAPS beams contribute interference
to the terrestrial user $U^{(T)}_k$. Unlike the single-beam case, where the
HAPS interference power follows a single shadowed-Rician distribution and an
exact outage expression can be obtained, the aggregate interference in the
multiple beam scenario is given by $I^{(H)}_k$ (as explained earlier),
which represents a sum of independent but non-identically distributed
shadowed Rician random variables. Since the exact PDF of such a summation is analytically intractable, the aggregate interference is approximated by a Gamma distribution using moment matching.
Accordingly,
$I^{(H)}_k\approx \Gamma\left(m_{I,k},\frac{\Omega_{I,k}}{m_{I,k}}\right)$,
where the shape parameter $m_{I,k}$ and the average interference power
$\Omega_{I,k}$ are obtained by matching the first two moments of the actual
interference distribution as
\begin{align}
\Omega_{I,k}&=\sum_{n\in\mathcal{N}_k}C_{k,n}P_{H,n}\left(\Omega_{k,n}+2b_{k,n}\right),\\
\sigma^2_{I,k}&=\sum_{n\in\mathcal{N}_k}B^2_{k,n}\left(\mathbb{E}\left[Z^2_{k,n}\right]-\left(\mathbb{E}\left[Z_{k,n}\right]\right)^2\right),
\end{align}
with $B_{k,n}=C_{k,n}P_{H,n}$. The equivalent Gamma shape parameter is therefore given by $m_{I,k}=\frac{\Omega^2_{I,k}}{\sigma^2_{I,k}}$, and the corresponding PDF is expressed as
\begin{align}\label{moment_PDF_gamma}
f_{I^{(H)}_k}(y)=\frac{\left(\frac{m_{I,k}}{\Omega_{I,k}}\right)^{m_{I,k}}}{\Gamma(m_{I,k})}y^{m_{I,k}-1}\exp\left(-\frac{m_{I,k}}{\Omega_{I,k}}y\right).
\end{align}

\begin{figure*}[t]
    \centering
   \begin{minipage}[b]{0.32\textwidth}
        \centering
       \includegraphics[width=\textwidth]{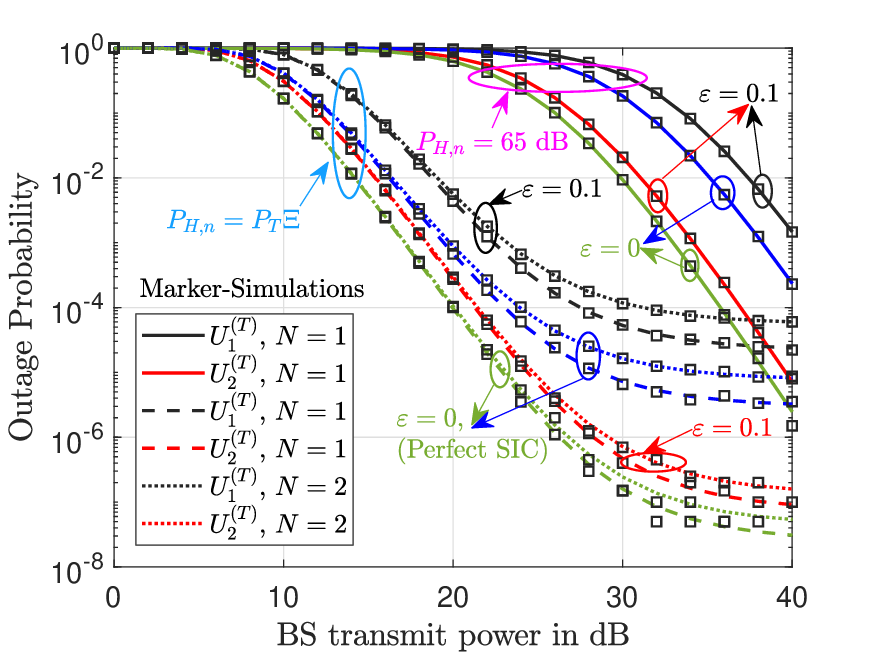}
       \caption{$\mathrm{P}_{\text{out},k}$ vs BS transmit power.}
        \label{fig1}
    \end{minipage}
     \begin{minipage}[b]{0.32\textwidth}
        \centering
       \includegraphics[width=\textwidth]{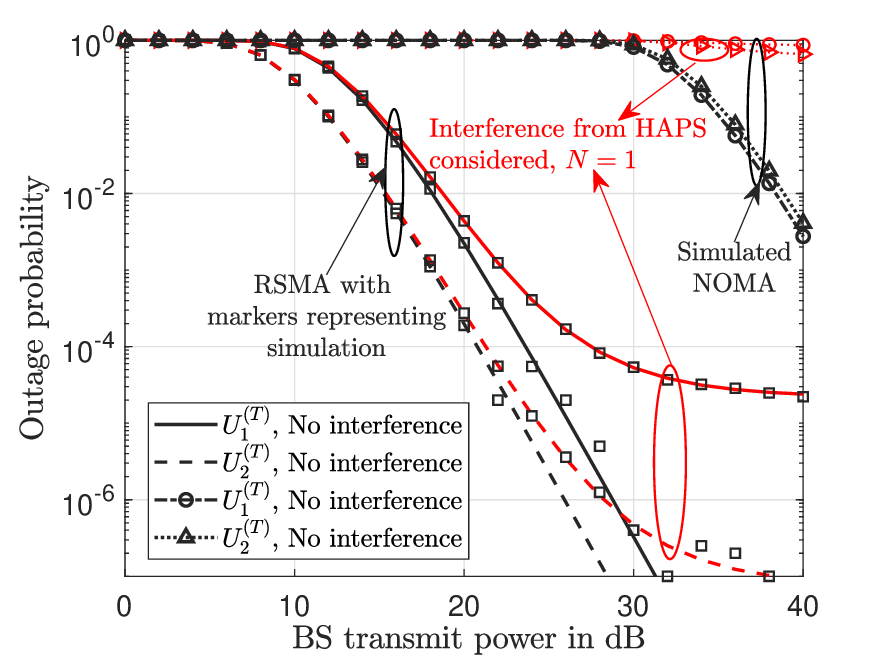}
        \caption{RSMA comparison with NOMA.}
        \label{fig2}
    \end{minipage}
    \begin{minipage}[b]{0.32\textwidth}
        \centering
       \includegraphics[width=\textwidth]{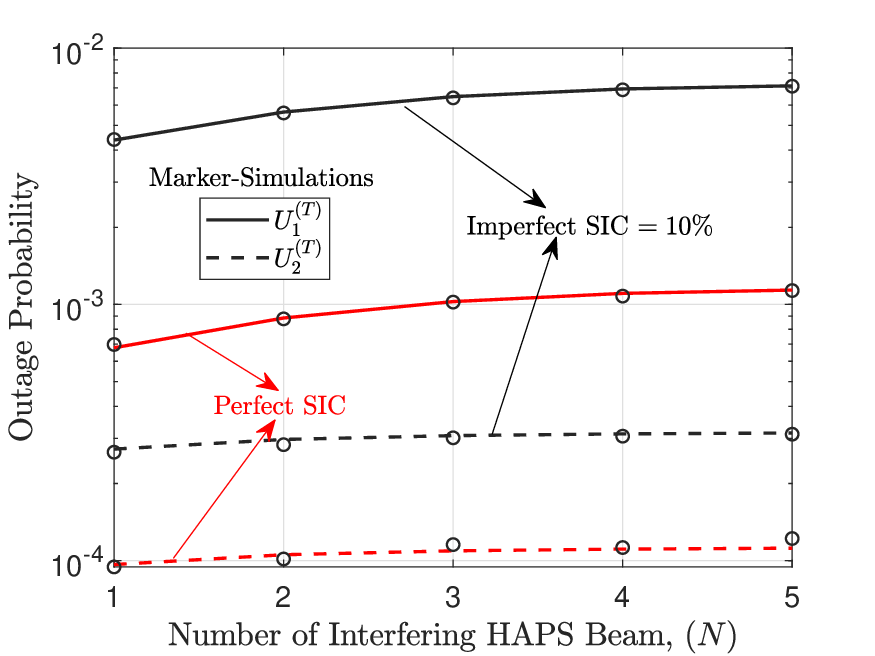}
        \caption{Effect of Interfering beam on $\mathrm{P}_{\text{out},k}$.}
        \label{fig3}
    \end{minipage} 
\end{figure*}

The outage probability considering multiple interfering HAPS beams is then
given in the following proposition.
\begin{proposition}
The outage probability of $U^{(T)}_k$ under multiple interfering HAPS beam is given by
\begin{align}\label{Multi_out}
      \mathrm{P}_{\text{out},k} &= 1-\sum\limits_{\ell=0}^{m_k-1}\sum\limits_{p=0}^{\ell}\binom{\ell}{p}\frac{\varphi_{k}^{\ell}}{\ell!}\left(\sigma^2_{T,k}\right)^{\ell-p}\exp\left(-\varphi_{k}\sigma^2_{T,k}\right)\nonumber\\
        &\times\left(\frac{m_{I,k}}{\Omega_{I,k}}\right)^{m_{I,k}} \frac{\Gamma\left(m_{I,k}+p\right)}{\Gamma(m_{I,k}) \left(\varphi_k+\frac{m_{I,k}}{\Omega_{I,k}}\right)^{m_{I,k}+p}},
\end{align}
where $m_{I,k}=\frac{\Omega^2_{I,k}}{\sigma^2_{I,k}}$, with rest of the parameters being already defined in Proposition~1.
\end{proposition}
\begin{proof}
    The reduction to $X_k\ge\varphi_k(I_k^{(H)}+\sigma_{T,k}^2)$ and \eqref{main_int} in the previous proof hold unchanged; only the interference PDF differs.
For multiple beam interference, $I_k^{(H)}$ is approximated by the moment-matched
Gamma PDF as given in \eqref{moment_PDF_gamma}.
Substituting this and $\bar F_{X_k}(\varphi_k(y+\sigma_{T,k}^2))$ into \eqref{main_int}, expanding $(y+\sigma_{T,k}^2)^\ell$, and using
$\int_0^\infty y^{m_{I,k}-1+p}\exp({-\hat{\eta} y})\; dy=\Gamma(m_{I,k}+p)/\hat{\eta}^{m_{I,k}+p}$
(valid for real $m_{I,k}>0$) with $\hat{\eta}=\varphi_k+m_{I,k}/\Omega_{I,k}$ yields
\eqref{Multi_out}.
\end{proof}

\begin{remark}
By substituting \eqref{Multi_out} in \eqref{throu_1}, $U^{(T)}_k$ throughput  considering multiple interfering HAPS beam can be obtained.
\end{remark}
\section{Numerical Results and Discussion}

This section validates the analytical outage expressions derived in Section~\ref{sec:analysis} against Monte Carlo simulations, and evaluates the outage performance of the TN under HAPS interference for varying numbers of interfering beams, as well as in the interference-free case. Although the derived expressions hold for $\mathcal{K}$ users, results are shown for two downlink users, $U^{(T)}_1$ and $U^{(T)}_2$, for clarity. the distance of $U^{(T)}_1$ and $U^{(T)}_2$ from BS is set at $250$~m and $265$~m, respectively and $\sigma^2_{T,k}=-90$~dB. Unless stated otherwise, all simulation parameters are listed in Table~\ref{sim_table}. 
\begin{table}[t]	\renewcommand{\arraystretch}{1.0}
		\centering
		\caption{ Simulation Parameters.}
		\label{sim_table}
			\resizebox{\columnwidth}{!}{\begin{tabular}{|l|l|l|l|l|l|}
			\hline
			Parameter         & Value         & Parameter & Value  & Parameter & Value  \\ \hline
			$\mathcal{H} $   &     $20 $~km   	&  $\psi_{1,1}  $   &     $40^{\circ}$ &  $\psi_{2,1} $   &     $30^{\circ}$ \\ \hline
			
			$G_H^{\max}$    &    $ 30$~dBi &	  $G_{U_k}^{(t)}$ &     $0$~dBi &   $R_{c}^{(T)}$ &     $0.85$  	 \\ \hline	

           $R_{p,1}^{(T)} $   &     $0.4 $ &	 $R_{p,2}^{(T)} $&     $0.85$ &   $m_1=m_2$ &     $4$  	 \\ \hline

           $m_{1,1}$   &     $7 $ &	  $\Omega_{1,1}$ &     $1.25$ &  $b_{1,1}$ &     $0.138$ 	 \\ \hline	
           
             $m_{2,1}$   &     $4 $ &	  $\Omega_{2,n}$ &     $1.0$ &  $b_{2,n}$ &     $0.1$ 	 \\ \hline	

$\zeta_{T,c}$   &     $0.56 $ &	  $\zeta_{T,1}$ &     $0.14$ &  $\zeta_{T,2}$ &     $0.3$ 	 \\ \hline

$G_{BS}$   &     $10$~dBi &	  $f_c$ &     $10$~GHz &  $A_m$ &     $25$~dB 	 \\ \hline

$L_r$   &     $0.05$~dB &	  $L_a$ &     $0.1$~dB &  $\varepsilon$ &     $0.1$ 	 \\ \hline
		\end{tabular}}
        \vspace{-1em}
	\end{table}
    
Fig.~\ref{fig1} illustrates the impact of HAPS interference on the outage performance of the TN users under different power-scaling conditions considering perfect and imperfect SIC. Three interference scenarios are considered: (i) a single HAPS beam with fixed transmit power, $P_{H,n}=65$~dB; (ii) a single HAPS beam whose transmit power scales with the TN transmit power according to $P_{H,n}=P^{(T)}\Xi$, where $\Xi=20$~dB; and (iii) two HAPS beams following the same scaling relationship with the second beam having values as $\psi_{1,2}=30^{\circ}$, $\psi_{2,2}=20^{\circ}$, $m_{1,2}=5$, $m_{2,2}=4$, $b_{1,2}=0.118$, $b_{2,2}=0.08$, $\Omega_{1,2}=1.05$, and $\Omega_{2,2}=0.8$. For Case I, the fixed interference dominates at low $P^{(T)}$, causing deep outage up to $P^{(T)}\leq 20$~dB. Beyond this point, rising $P^{(T)}$ overcomes the fixed interference and outage decreases monotonically.

In Cases~II and III, interference grows jointly with $P^{(T)}$, therefore, the outage saturates to a floor regardless of the transmit power, with the floor rising further in Case~III due to the second interfering beam. This contrast shows that HAPS interference impact depends not on its absolute power but on whether it scales with the desired TN signal. Fixed interference is eventually overcome, while proportional interference imposes a persistent limitation. Lastly, the analytical and simulated curves agree closely throughout, validating the expressions derived in Propositions~1 and~2.

Fig.~\ref{fig2} compares RSMA and NOMA TN outage, with and without HAPS interference considering imperfect SIC. NOMA power coefficients are $0.42$ and $0.58$, respectively, with target rates of $0.83$ and $1.25$ and decoding order $U^{(T)}_2>U^{(T)}_1$. For a fair comparison with RSMA, the NOMA power allocation is obtained by distributing the RSMA common-stream power allocation equally between the two users and adding the resulting portions to their respective private-stream powers.
At low $P^{(T)}$, RSMA with and without HAPS interference nearly coincide, as the system remains noise-limited. As $P^{(T)}$ increases, the interference-free curve keeps improving, while the interference case saturates once
$P_{H,n}$ scales with $P^{(T)}$ by factor $\Xi$ and the system enters the interference-limited regime, the same mechanism as Fig.~\ref{fig1}. RSMA maintains lower outage than NOMA throughout, both with and without interference, indicating that its flexible stream structure manages HAPS coexistence more robustly than NOMA's fixed decoding order. Lastly, the interference-free RSMA simulation matches the analytical curve, validating Corollary~1.

Fig.~\ref{fig3} shows the outage probability versus the number of interfering HAPS beams $N$, for perfect and imperfect SIC, with $P^{(T)}=20$~dB and $P_{H,n}=P^{(T)}\Xi$ with $\Xi=20$~dB. The per-beam elevation angles are $\psi_{1,n}\in\{40^\circ,30^\circ,25^\circ,20^\circ,15^\circ\}$ and $\psi_{2,n}\in\{30^\circ,20^\circ,15^\circ,12^\circ,10^\circ\}$ with corresponding shadowed-Rician parameters $m_{1,n}\in\{7,5,5,4,4\}$, $b_{1,n}\in\{0.138,0.118,0.11,0.105,0.1\}$, $\Omega_{1,n}\in\{1.25,1.05,0.95,0.85,0.75\}$ for $U^{(T)}_1$, and $m_{2,n}\in\{4,4,3,3,3\}$, $b_{2,n}\in\{0.100,0.08,0.075,0.07,0.065\}$, $\Omega_{2,n}\in\{1.0,0.8,0.7,0.6,0.55\}$ for $U^{(T)}_2$. Imperfect SIC yields outage levels of approximately $10^{-2.5}$ and $10^{-3}$ for $U^{(T)}_1$ and $U^{(T)}_2$, respectively, while perfect SIC improves these to about $10^{-3.5}$ and $10^{-4}$ due to the removal of residual interference. The impact of increasing $N$ is more pronounced under imperfect SIC, where outage increases mainly from $N=1$ to $N=3$ and becomes less sensitive to further increases in $N$. In contrast, perfect SIC shows a comparatively weaker sensitivity to increasing $N$ over the considered range. This indicates that residual interference leaves the users more susceptible to additional HAPS interference, whereas perfect SIC provides a larger SINR margin and therefore reduces sensitivity to the number of interfering beams.

\begin{figure}[t]
\centering
\includegraphics[width=0.48\columnwidth]{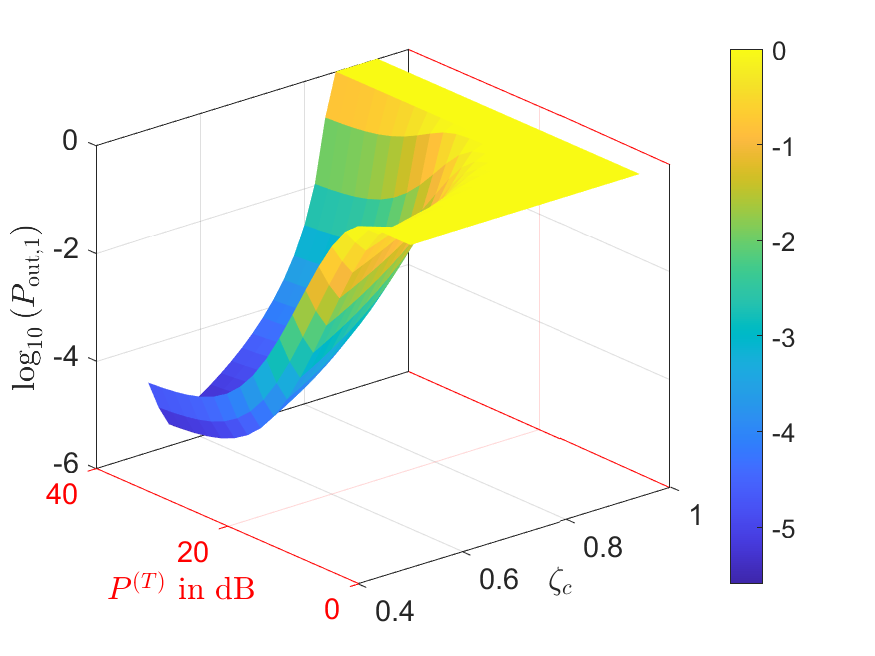}
\hfill
\includegraphics[width=0.48\columnwidth]{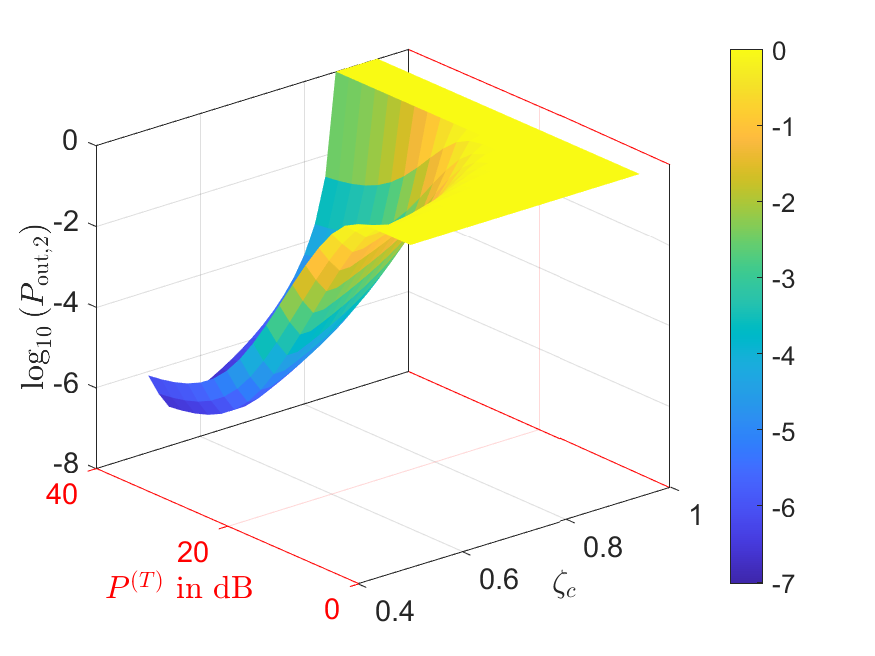}
\caption{Effect of $P^{(T)}$ and $\zeta_c^{(T)}$ on outage probability.}
\label{fig23}
\vspace{-1.2em}
\end{figure}

Fig.~\ref{fig23} show the outage probability of $U^{(T)}_1$
and $U^{(T)}_2$, respectively, as a joint function of $P^{(T)}$ and the
common-stream power fraction $\zeta_c^{(T)}$, for $N=3$ interfering HAPS
beams (value corresponds to  the first three beams of the sets used in Fig.~\ref{fig3}), $\epsilon=0.1$, and $P_{H,n}=P^{(T)}\Xi$ with $\Xi=20$~dB. Both
surfaces exhibit a valley bounded by two saturated regions: at low
$P^{(T)}$, outage approaches unity for all $\zeta_c^{(T)}$ since the
desired signal is too weak regardless of the power split, whereas at
$\zeta_c^{(T)}\gtrsim0.84$ outage saturates due to a hard infeasibility
condition, insufficient private power to meet the rate threshold,
irrespective of $P^{(T)}$. Notably, the minimizing $\zeta_c^{(T)}$ remains
close to $0.54$ across the entire power range, indicating that the optimal
split is largely power-invariant for this scenario and closely matches the
operating point $\zeta_c^{(T)}=0.56$ used in Figs.~\ref{fig1}-\ref{fig3},
confirming that those results were generated near the outage-optimal
allocation.

\section{Conclusion}
This paper studied the coexistence of RSMA-based TN with HAPS-based NTN interference, jointly modeling Nakagami-$m$ fading on the desired TN link and shadowed-Rician fading on the aggregate multibeam HAPS interference. Closed-form outage expressions were derived and validated against Monte Carlo simulations under perfect and imperfect SIC. Our results show that the impact of HAPS interference on TN outage performance is governed not by its absolute power but by how it scales relative to the TN transmit power: fixed-power interference is eventually overcome, whereas interference that scales jointly with TN transmit power produces an irreducible outage floor. We further showed that RSMA consistently outperforms NOMA under HAPS interference, while residual SIC imperfections amplify sensitivity to the number of interfering beams, underscoring the importance of effective SIC in interference-limited TN-NTN coexistence. These findings offer practical design guidance for HAPS power control and RSMA resource allocation in future 6G TN-NTN integrated networks.

\bibliographystyle{IEEEtran_renamed}
	\bibliography{refer}
\end{document}